\documentclass[a4paper]{article}
\usepackage{amsfonts,amssymb,amsmath,amsthm,cite}
\usepackage{ucs} 
\usepackage[utf8x]{inputenc}
\usepackage{graphicx}
\usepackage[english]{babel}
\usepackage{slashed}
\usepackage{textcomp}
\usepackage{bm}

\evensidemargin=\oddsidemargin
\newtheorem{theorem}{Theorem}
\newtheorem{lemma}{Lemma}

\begin{document}
\begin{center}
	{\Large\textbf{Explicit Cutoff Regularization in Coordinate Representation}}
	\vspace{0.5cm}
	
	{\large Aleksandr~V.~Ivanov}
	
	\vspace{0.5cm}
	
	{\it St. Petersburg Department of Steklov Mathematical Institute of
		Russian Academy of Sciences,}\\{\it 27 Fontanka, St. Petersburg 191023, Russia}\\
	{\it Leonhard Euler International Mathematical Institute, 10 Pesochnaya nab.,}\\
	{\it St. Petersburg 197022, Russia}\\
	{\it E-mail: regul1@mail.ru}
\end{center}
\vskip 10mm
\date{\vskip 20mm}

\begin{abstract}
In this paper, we study a special type of cutoff regularization in the coordinate representation. We show how this approach unites such concepts and properties as an explicit cut, a spectral representation, a homogenization, and a covariance. Besides that, we present new formulae to work with the regularization and give additional calculations of the infrared asymptotics for some regularized Green's functions, appearing in the pure four-dimensional Yang--Mills theory and in the standard two-dimensional Sigma-model.
\end{abstract}
\normalsize
\tableofcontents


\section{Introduction}
Divergent integrals appear in physical and mathematical calculations quite frequently. For example, the perturbative quantum field theory \cite{9,10} is based on studying of Feynman diagrams, which contains divergent quantities of different types due to a "bad" infrared behaviour of Green's functions (propagators). Fortunately, the renormalization theory \cite{6,7} was created that allows us to obtain physically meaningful results from the divergencies. However, such approach depends on a regularization scheme, see \cite{105}. There are a lot of ways to perform it, but we are interested in only one of them, so-called a cutoff regularization.

In this work we are going to discuss the cutoff regularization of a special type. We follow the papers \cite{Ivanov-Kharuk-2019,Ivanov-Kharuk-2020,Ivanov-Kharuk-2022} and introduce an explicit cut in the coordinate representation. However, we want the regularization to possess and unite some additional properties, such as an explicit cutoff procedure, a spectral representation, a homogenization, and a covariance. All these concepts and conditions can be satisfied. In the next section we will explain our point of view in detail. Description of some other cutoff regularizations can be found in \cite{w6,w7,w8}.

The paper has the following structure. First of all, in Section \ref{1:sec:aprd}, we describe the main ideas of our approach and explain some conditions, which we want to satisfy. Then, in Section \ref{1:sec:prst}, we make the first steps to realize the project. In the section a one-dimensional case is studied. In Section \ref{1:sec:gener}, we calculate necessary spectral functions to give the formulation for an arbitrary dimension value. In Section \ref{1:sec:fur}, we consider some generalizations and study cases, applicable to concrete models, such as the pure four-dimensional Yang--Mills theory and the standard two-dimensional Sigma-model, see Theorems \ref{1:cal:th1} and \ref{1:cal:th2}. Some additional remarks and discussion are presented in Section \ref{1:sec:conc}.

\section{Approach description}
\label{1:sec:aprd}
Let us consider a Laplace-type operator $L$ on $\mathbb{R}^d$, where $d\in\mathbb{N}$, with smooth coefficients. Then, we introduce eigenvalues $\lambda$ and eigenfunctions
$\phi_\lambda$ for the operator, such that $L(x)\phi_\lambda(x)=\lambda\phi_\lambda(x)$ for all $x\in\mathbb{R}^d$. Besides this we assume that all $\lambda>0$. Of course, we have implied the presence of suitable boundary conditions and a space of solutions, so that the problem is correct.
We do not discuss the spectral problem in detail here, because we want to explain the main idea of our approach. However, we will assume that the Green's function $G$ exists and its kernel has the following representation
\begin{equation}\label{1:int2}
G(x,y)=
\int_{\mathbb{R}_+}d\mu(\lambda)\,
\phi^{\phantom{*}}_\lambda(x)
\frac{1}{\lambda^2}
\phi^*_\lambda(y),\,\,\,x,y\in\mathbb{R}^d,
\end{equation}
where the integration is performed by using a measure $d\mu(\lambda)$ that corresponds to the operator $L$, see \cite{Birman}. The star denotes the hermitian conjugation. Let us note that if the last measure contains only a discrete spectrum, then we can replace the integration by a sum over all eigenvalues with respect of their multiplicities.

As it was mentioned in the introduction, Green's functions play an important role in the perturbative quantum field theory, because they are constructive blocks for Feynman diagrams. So, they can appear under an integration in different nonlinear combinations. This fact leads to the appearance of divergences, which permeate modern theoretical physics. For example, on the four-dimensional space, the right hand side of formula (\ref{1:int2}) has the asymptotics $(1+o(1))/(4\pi|x-y|^2)$ in the range $x\sim y$. Therefore, the Green's function has the singularity at $x=y$. Such situation arises in the Yang--Mills theory, for instance.

To avoid the singularities and, moreover, divergences, we need to introduce some type of regularization. There are a lot of ways to perform this, but we are interested in a cutoff regularization in the coordinate representation. Let us describe several conditions we are going to take into account.

$(i)$ \textbf{Explicit cutoff.} First of all, the regularization should have a clear cutting procedure. In the introduction, we have cited some ways. However, now we suggest to use the one studied and successfully used in the works \cite{Ivanov-Kharuk-2019,Ivanov-Kharuk-2020,Ivanov-Kharuk-2022}. In the next sections we give precise definitions, whereas here we want to consider one example. Returning to the four-dimensional case mentioned above, the procedure means that the main term of the Green's function has the following deformation
\begin{equation}\label{1:int3}
|x-y|_{\Lambda}=\begin{cases}
|x-y|,&|x-y|>1/\Lambda;\\
\,\,\,1/\Lambda,&|x-y|\leqslant 1/\Lambda,
\end{cases}
\end{equation}
where the limit $\Lambda\to+\infty$ removes the regularization. The non-leading terms also have deformations, but they contain ambiguities that should be controlled by additional conditions, described below. These ambiguities appear due to the ability to add to the deformed Green's function an additional function $g_\Lambda(x,y)$, such that $Lg_\Lambda\to0$ for $\Lambda\to+\infty$ in the sense of generalized functions \cite{Gelfand-1964,Vladimirov-2002}.

$(ii)$ \textbf{Spectral representation.} The procedure should have an explicit spectral meaning, because we want to know how spectral functions of the operator $L$ are deformed. The most expected way is to present the regularization by an operator of integration $\mathfrak{J}$, the kernel $\mathfrak{J}(x,y)$ of which is equal to the following integral
\begin{equation}\label{1:int1}
\mathfrak{J}_\Lambda(x,y)=\int_{\mathbb{R}_+}d\mu(\lambda)\,
\phi^{\phantom{*}}_\lambda(x)
\rho(\lambda/\Lambda)\phi^*_\lambda(y),
\end{equation}
where the function $\rho(\cdot)$ has properties, sufficient to work with the representation. It is clear that we need the equality $\rho(0)=1$ holds, because the limit transition $\Lambda\to+\infty$ removes the regularization.

$(iii)$ \textbf{Homogenization.} Another question that arises during the construction of the cutoff regularization concerns the possibility of its representation by a "classical" integration operator. In addition, we would like this integration operator, if it exists, to be a homogenization operator by a sphere near the point under study. For example, in the four-dimensional case mentioned above, this question can be reformulated for the first order as follows: is there such a smooth kernel $\omega(y)$ that the relations hold
\begin{equation}\label{1:int4}
\frac{1}{4\pi|x|^2_\Lambda}=\frac{1}{S_3}\int_{\mathbb{S}_{\phantom{1}}^3}d\tilde{\sigma}(y)\,
\frac{\omega(y)}{4\pi|y/\Lambda+x|^2},\,\,\,
\frac{1}{S_3}\int_{\mathbb{S}_{\phantom{1}}^3}d\tilde{\sigma}(y)\omega(y)=1,
\end{equation}
where $\mathbb{S}_{\phantom{1}}^3$ is the unit sphere in $\mathbb{R}^4$ at the origin, $S_3$ is its surface area, and $d\tilde{\sigma}(y)$ is a measure on the sphere? The answer is positive, as we will see it below.

$(iv)$ \textbf{Covariance.} Here, we want to say that the kernel $\mathfrak{J}_\Lambda(x,y)$ should change covariantly with respect to gauge transformations, if, of course, such type of transformations is applicable to the operator $L$. Actually, this follows from decomposition (\ref{1:int1}) and, therefore, the kernel $\mathfrak{J}_\Lambda(x,y)$ inherits properties of the operator $L(x)$. Indeed, from the change
\begin{equation}\label{1:int5}
L(x)\to U(x)L(x)U^{*}(x)\,\,\,\mbox{we obtain}\,\,\,
\mathfrak{J}_\Lambda(x,y)\to U(x)\mathfrak{J}_\Lambda(x,y)U^{*}(y)
\end{equation}
due to the appropriate change of the eigenfunctions $\phi_\lambda(x)\to U(x)\phi_\lambda(x)$, where $U(x)$ is a smooth element of a transformation group.

As we know, Lagrangian densities for the standard field theories are local, because they contain a finite number of derivatives. At the same time the property $(ii)$ can lead to losing the locality, while the condition $(iii)$ allows us to keep the quasi-locality. Indeed, we can replace the regularization in form (\ref{1:int4}) from the operator to the field, and, hence, we obtain the local homogenization.

\section{First attempts}
\label{1:sec:prst}
In this section we consider and demonstrate one simple example, which shows the main idea of the approach. Let us introduce the standard Laplace operator $A(s)=-\partial_s\partial_s$ on the real axis $\mathbb{R}$, which has the twofold continuous spectrum. It means that the equation $A(s)\psi(s)=\lambda^2\psi(s)$ has two oscillating solutions $\exp(\pm i\lambda s)$ for all $\lambda\in\mathbb{R}_+$, which form a kernel for the Fourier transform and, therefore, lead to the momentum representation, see \cite{Vladimirov-1981}.

It is known that the equation $A(s)g(s)=\delta(s)$ has a lot of solutions, which are called fundamental solutions. Now we are interested in $g(s)$ of a special type, namely $G(|s|)=-|s|/2$. It is easy to verify that other fundamental solutions can be obtained from $G(|s|)$ by adding so-called "zero modes" in the form $g(s)=G(|s|)+as+b$, where $a,b\in\mathbb{R}$.

Actually, the function $|s|$ is not very "good" in terms of further investigation of the function $|s|^{-k}$, where $k\in\mathbb{N}$, because it goes to $+\infty$ for $s\to\pm0$. One of the methods to avoid such type of singularities is to apply a cutoff regularization. As it was mentioned above, there are a lot of ways to introduce the one, but the simplest approach consists in the following change
\begin{equation}\label{1:cut1}
|s|\to|s|_{\Lambda}=
\begin{cases}
\,\,|s|,&|s|>1/\Lambda;\\
1/\Lambda,&|s|\leqslant 1/\Lambda,
\end{cases}
\end{equation}
where $\Lambda>0$ is an auxiliary regularization parameter that is quite large.

In this case the solution $G(|s|)$ is transformed into the regularized one $G(|s|_{\Lambda})$, which satisfies the following limit transition $G(|s|_{\Lambda})\to G(|s|)$ for $\Lambda\to+\infty$ in the sense of generalized functions \cite{Gelfand-1964, Vladimirov-2002}.

Let us find the deformation of the second derivative
\begin{equation}\label{1:cut2}
A(s)G(|s|_{\Lambda})=\frac{\partial_s}{2}
\begin{cases}
\mathrm{sign}(s),&|s|>1/\Lambda\\
\,\,\,\,\,\,\,0,&|s|\leqslant 1/\Lambda
\end{cases}=\frac{1}{2}\sum_{\sigma=\pm 1}\delta(s-\sigma/\Lambda)=
\frac{1}{S_0}\int_{\mathbb{S}_{\phantom{1}}^{0}}d\sigma(r)\,\delta(s-r/\Lambda),
\end{equation}
where the surface area $S_0=2$ for the $0$-dimensional sphere $\mathbb{S}_{\phantom{1}}^{0}=\{-1,1\}$, and the measure is defined by the last equation. The last calculation means that the $\delta$-functional gets the deformation, and it is equal to the average value of all the $\delta$-functionals on the sphere. We can say that we have obtained some type of homogenization of the standard $\delta$-functional.

This type of homogenization is very remarkable and it plays an important role in the approach. Let us pay attention to the fact that on the real axis, we have the standard Fourier transform. So, a smooth function $f\in C_0^{\infty}(\mathbb{R})$ with a compact support can be represented by its momentum image
\begin{equation}\label{1:cut3}
f(s)=\frac{1}{2\pi}\int_{\mathbb{R}}dt\,e^{ist}\hat{f}(t),\,\,\,\mbox{where}\,\,\,\,
\hat{f}(t)=\int_{\mathbb{R}}ds\,e^{-ist}f(s),
\end{equation}
and the hat denotes the Fourier transformed version.
In addition let us note the presence of the following relation
\begin{equation}\label{1:cut4}
\frac{1}{2\pi}\int_{\mathbb{R}}dr\,e^{i(s-t)r}=\delta(s-t).
\end{equation}

Hence, the homogenization of the function can be reformulated as the homogenization of the exponential. Therefore, we can write the following relation
\begin{equation}\label{1:cut5}
\frac{1}{S_0}\int_{\mathbb{S}_{\phantom{1}}^{0}}d\sigma(r)\,f(s-r/\Lambda)=
\frac{1}{2\pi}\int_{\mathbb{R}}dt\,e^{ist}\hat{f}(t)\rho(t/\Lambda),
\end{equation}
where
\begin{equation}\label{1:cut6}
\rho(t/\Lambda)=\frac{1}{S_0}\int_{\mathbb{S}_{\phantom{1}}^{0}}d\sigma(r)\,e^{itr/\Lambda}=
\frac{1}{2}\sum_{\sigma=\pm 1}e^{it\sigma/\Lambda}=\cos(t/\Lambda).
\end{equation}
We see that the function $\rho(\cdot)$ gives additional ascillations, which perform the regularization. Using the limit $\Lambda\to+\infty$, these oscillations can be excluded. In high dimensions, the function gets an additional decreasing at the infinity.

On the other hand, we can use the Taylor expansion of the function $\cos(\cdot)$ and take it from under the integration in the form of the derivatives. Then, we get
\begin{equation}\label{1:cut7}
\frac{1}{S_0}\int_{\mathbb{S}_{\phantom{1}}^{0}}d\sigma(r)\,f(s-r/\lambda)=
\sum_{n=0}^{+\infty}\frac{(-1)^n}{(2n)!}\Big(A(s)/\Lambda^2\Big)^nf(s)=\rho\Big(\sqrt{A(s)}/\Lambda\Big)f(s),
\end{equation}
where we have used the operator-valued function, which is understood as the corresponding Taylor series.

To obtain a connection between functions (\ref{1:cut1}) and (\ref{1:cut6}) we need to make one more step. First of all, let us remember the Fourier transform of $G(|s|)$
\begin{equation}\label{1:cut8}
\int_{\mathbb{R}}ds\,e^{ist}G(|s|)=\frac{1}{t^2}=\hat{G}(|t|).
\end{equation}

Then, using the representation $G(|s|_\Lambda)=G(|s|)+\big(G(|s|_\Lambda)-G(|s|)\big)$ and the fact that the difference $G(|s|_\Lambda)-G(|s|)=0$ for all $|s|\geqslant 1/\Lambda$, we get the following chain of equalities
\begin{equation}\label{1:cut9}
\int_{\mathbb{R}}ds\,e^{ist}G(|s|_\Lambda)=\frac{1}{t^2}+\int\limits_{-1/\Lambda}^{1/\Lambda}ds\,e^{ist}
\bigg(\frac{|s|}{2}-\frac{1}{2\Lambda}\bigg)=\frac{\rho(t/\Lambda)}{t^2}
=\rho(t/\Lambda)\hat{G}(|t|).
\end{equation}

Finally, applying the inverse Fouries transform and formulae (\ref{1:cut4}) and (\ref{1:cut5}), we have
\begin{equation}\label{1:cut10}
G(|s|_\Lambda)=\frac{1}{2\pi}\int_{\mathbb{R}}dt\,e^{ist}\hat{G}(|t|)\rho(t/\Lambda)=
\frac{1}{S_0}\int_{\mathbb{S}_{\phantom{1}}^{0}}d\sigma(r)\,G(|s-r/\Lambda|).
\end{equation}

After all the calculations, we have obtained two very interesting and remarkable relations. On one hand, we have shown that our cutoff regularization (\ref{1:cut1}) can be achieved by using the operator of homogenization in the form (\ref{1:cut5}). And on the other hand, such type of the regularization has an explicit spectral meaning, because from the spectral decomposition of formula (\ref{1:cut8}) we see that the eigenvalue $1/\lambda^2$ goes to $\rho(\lambda/\Lambda)/\lambda^2$.
Additionally, let us note that the limit $\Lambda\to+\infty$ removes the regularization, because $\rho(0)=1$.

In the next sections, we describe our approach for all dimensions and more intricate Laplace-type operators. Besides that we make some additional calculations for several special operators.

\section{Generalization}
\label{1:sec:gener}

To move on, let us introduce some additional convenient notations. First of all, we abandon the restriction on the dimension, so we study $d\in\mathbb{N}$. Then, considering $x\in\mathbb{R}^d$, we can define an auxiliary function $G(|x|)$ as
\begin{equation}\label{Green}
G(|x|)=
\begin{cases}
\,\,\,\,\,\,\,\,\,\,\,\,\,-|x|/2,&d=1;\\
\,\,\,\,\,\,-\ln(|x|)/2\pi,&d=2;\\
|x|^{2-d}/(d-2)S_{d-1},&d\geqslant 3,
\end{cases}
\end{equation}
where $S_{d-1}=2\pi^{d/2}/\Gamma(d/2)$ is the surface area of $\mathbb{S}^{d-1}$, the $(d-1)$-dimensional unit sphere, centered at the origin. 
Also, we have used $|x|=\sqrt{x_\mu x^\mu}$ with the corresponding summation by the repeating indices.
Of course, the last function $G(\cdot)$ depends on the dimension value, which will be omitted as a rule, because this does not cause confusion. Let us note that function (\ref{Green}) satisfies the very well known equality $-\partial_{x_\mu}\partial_{x^\mu}G(|x-y|)=\delta(x-y)$, where $x,y\in\mathbb{R}^d$.

In the next calculations, we repeat all steps from the previous section. The generalization means that we abandon the restriction on the dimension parameter, though it takes only positive integer values.

Following the main idea, we need to find the Fourier transformation for the regularized, see (\ref{1:int3}), Green's function $G(|x|_\Lambda)$.

\begin{lemma}\label{fur1}
Let all the conditions described above be valid. Also,
let $x,k\in\mathbb{R}^d$, $s=|x|$, $t=|k|$, and $s_{\Lambda}$ is the corresponding deformation according to formula (\ref{1:cut1}), where $\Lambda>0$. 
Then, we have the following relation
\begin{equation}\label{plusandmin4}
\int_{\mathbb{R}^d}d^dx\,e^{ik_\mu x^\mu}G(s_{\Lambda})=
\frac{1}{S_{d-1}t^2}
\int_{\mathbb{S}^{d-1}_{\phantom{1}}}d\sigma(x)\,e^{ik_\mu x^\mu/\Lambda}.
\end{equation}
\end{lemma}
\begin{proof} Let us proceed in stages. The situation $d=1$ was studied in the previous section. In the case $d\geqslant 3$ we are going to add and subtract the function $1/s^{d-2}$ to the integrand in the region $s\Lambda<1$. Then, using the Fourier transform for the standard Green's function, we obtain
\begin{equation}\label{plusandmin2}
\frac{1}{(d-2)S_{d-1}}
\int_{\mathbb{R}^d}d^dx\,\frac{e^{ik_\mu x^\mu}}{s^{d-2}_{\Lambda}}=
\frac{1}{t^2}+\frac{1}{(d-2)S_{d-1}}\int_{\mathbb{B}_{1/\Lambda}^d}d^dx\,e^{ik_\mu x^\mu}\bigg(\Lambda^{d-2}-\frac{1}{s^{d-2}}\bigg),
\end{equation}
where $\mathbb{B}_{1/\Lambda}^{d}$ is the closed $d$-dimensional ball with the radius $1/\Lambda$, centered at the origin.

Further, let us transform the second term in the last equality. For that we represent the exponential in the following form $e^{ik_\mu x^\mu}=-\partial_{x_\mu}\partial_{x^\mu}e^{ik_\mu x^\mu}/t^2$ and then integrate by parts twice
\begin{equation}\label{plusandmin3}
\int_{\mathbb{B}_{1/\Lambda}^{d}}d^dx\,\Big(\partial_{x_\mu}\partial_{x^\mu}e^{ik_\mu x^\mu}\Big)\bigg(\Lambda^{d-2}-\frac{1}{s^{d-2}}\bigg)=
(d-2)S_{d-1}-(d-2)\int_{\mathbb{S}^{d-1}_{\phantom{1}}}d\sigma(x)\,e^{ik_\mu x^\mu/\Lambda},
\end{equation}
where we have used the equality $\partial_{x_\mu}\partial_{x^\mu}s^{2-d}=-(d-2)S_{d-1}\delta(x)$.

Hence, substituting relation (\ref{plusandmin3}) in formula (\ref{plusandmin2}) and making the change of variables $x^\mu\to x^\mu/\Lambda$, we obtain the statement of the lemma for $d>2$.

The equality for $d=2$ can be achieved in the same manner. We note separately that all the calculations are performed in the sense of generalized functions on the space of functions with the compact support.
\end{proof}

We see from the last lemma, that the Fourier transformation for $G(|x|_\Lambda)$ contains an additional factor, compared to the non-regularized function. This factor has an integral form. Fortunately, it can be calculated explicitly.

\begin{lemma}
\label{fur3}
Let $k\in\mathbb{R}^d$, $t=|k|$, and $\Lambda>0$. Then,
under the conditions described above, the relation holds 
\begin{equation}\label{ffa1}
\rho(|t|/\Lambda)=
\frac{1}{S_{d-1}}\int_{\mathbb{S}^{d-1}_{\phantom{1}}}d\sigma(y)\,e^{ik_\mu y^\mu/\Lambda}=
\Gamma(d/2)\bigg(\frac{|t|}{2\Lambda}\bigg)^{1-d/2}J_{d/2-1}(|t|/\Lambda).
\end{equation}
\end{lemma}
\begin{proof}
To obtain the statement of the lemma, we need to derive one auxiliary relation. Let
the hat denotes, that the corresponding vector is unit, so $\hat{y}^\mu=y^\mu/|y|$ for all $y\in\mathbb{R}^d$. Let $m=2n$ with $n\in\mathbb{N}\cup\{0\}$, $c=2^{n+d/2}\Gamma(n+d/2)$, and $A_{\mu_1\ldots\mu_m}$ is a symmetric tensor. Then, analogously to \cite{DIS-2020}, we have the following chain of equalities
\begin{align}\label{ffa0}
\int_{\mathbb{S}^{d-1}_{\phantom{1}}}d\sigma(y)\,\hat{y}^{\mu_1\ldots\mu_m}A_{\mu_1\ldots\mu_m}&=
\frac{2}{c}\int_{\mathbb{R}_+}dr\,e^{-r^2/2}r^{m+d-1}
\int_{\mathbb{S}^{d-1}_{\phantom{1}}}d\sigma(y)\,\hat{y}^{\mu_1\ldots\mu_m}A_{\mu_1\ldots\mu_m}\\\label{ffa2}&=
\frac{2}{c}\int_{\mathbb{R}^d}d^dy\,e^{-|y|^2/2}
y^{\mu_1\ldots\mu_m}A_{\mu_1\ldots\mu_m}\\\label{ffa3}&=
\frac{2(2\pi)^{d/2}}{c}
\partial_{z_{\mu_1}}\cdot\ldots\cdot\partial_{z_{\mu_m}}
A_{\mu_1\ldots\mu_m}e^{|z|^2/2}\Big|_{z=0}\\\label{ffa4}&=
\frac{2(2\pi)^{d/2}}{c}\frac{(2n)!}{2^nn!}
A_{\mu_1\mu_1\ldots\mu_n\mu_n}.
\end{align}
Hence, using the Taylor expansion for the exponential on the left hand side of (\ref{ffa1}), applying the last relation to the each term, and then summing the resulting series by using
\begin{equation}\label{ffa5}
\sum_{n=0}^{+\infty}\frac{(-r^2/4)^n}{n!\Gamma(n+d/2)}=(r/2)^{1-d/2}J_{d/2-1}(r),
\end{equation}
where $r\geqslant0$, we obtain the statement of the lemma.
\end{proof}

From the last lemma it follows that the factor $\rho(\cdot)$ is presented by the oscillating function, which depends on the dimension parameter. 
We have omitted this parameter in definition (\ref{ffa1}), because this does not lead to confusion.
For example, in the case $d=1$, we get the known result from formula (\ref{1:cut6}). Indeed, we get
\begin{equation}\label{1:cala:14}
\Gamma(1/2)\bigg(\frac{|t|}{2\Lambda}\bigg)^{1/2}J_{-1/2}(|t|/\Lambda)=
\cos(|t|/\Lambda).
\end{equation}
So, the obtained result  is consistent with the previous one.

Finally, we can derive a general representation formula using an integration by $(d-1)$-dimensional sphere. Such formula gives a connection of the regularization with the homogenization, see Section \ref{1:sec:aprd}.

\begin{lemma}\label{fur2}
Let $x\in\mathbb{R}^d$, $s=|x|$, $\Lambda>0$, and $s_\Lambda$ is the corresponding deformation from (\ref{1:cut1}).
Then, under the conditions described above, we have the following additional representation for (\ref{Green})
\begin{equation}\label{fff0}
G(s_{\Lambda})=\frac{1}{S_{d-1}}
\int_{\mathbb{S}^{d-1}_{\phantom{1}}}d\sigma(y)\,G(|x-y/\Lambda|).
\end{equation}
\end{lemma}
\begin{proof}
To verify the equality we need to apply the inverse Fourier transform to the both sides of (\ref{plusandmin4}). Hence, the left hand side gives (\ref{Green}), while for the right hand side we can write out the following chain of equalities
\begin{align}
\label{fff1}
\frac{\partial_{x_\mu}\partial_{x^\mu}}{(4\pi)^{d}}\int_{\mathbb{R}^d}d^dk\,
\frac{e^{ik_\mu x^\mu}}{S_{d-1}|k|^2}
\int_{\mathbb{S}^{d-1}_{\phantom{1}}}d\sigma(y)\,e^{ik_\mu y^\mu/\Lambda}&=-
\frac{1}{S_{d-1}}
\int_{\mathbb{S}^{d-1}_{\phantom{1}}}d\sigma(y)\,\delta(x-y/\Lambda)\\&=
\frac{\partial_{x_\mu}\partial_{x^\mu}}{S_{d-1}}
\int_{\mathbb{S}^{d-1}_{\phantom{1}}}d\sigma(y)\,G(|x-y/\Lambda|),
\end{align}
from which the statement of the lemma follows. Note that on the last step we have excluded the operator $\partial_{x_\mu}\partial_{x^\mu}$ and then we have fixed local "zero modes" with the use of the obvious value at the origin $G(|0|_\Lambda)=G(1/\Lambda)$.
\end{proof}

The last three lemmas lead to the main result discussed above. On the one hand, we have the connection between the explicit cutoff regularization of the Green's function and the homogenization with the use of the $(d-1)$-dimensional sphere, see formulae (\ref{1:int3}) and (\ref{fff0}). On the other hand, we have obtained the relation between the regularized Green's function and the corresponding momentum (spectral) representation. All the equations can be written in the form for $x\in\mathbb{R}^d$
\begin{equation}\label{1:cala:15}
G(|x|_\Lambda)=\frac{1}{(2\pi)^d}\int_{\mathbb{R}^d}d^dy\,e^{ix_\mu y^\mu}\hat{G}(|y|)\rho(|y|/\Lambda)=
\frac{1}{S_{d-1}}\int_{\mathbb{S}_{\phantom{1}}^{d-1}}d\sigma(y)\,G(|x-y/\Lambda|),
\end{equation}
which is actually the generalization of the one-dimensional case, see (\ref{1:cut10}), to an arbitrary dimension. Additionally, we note that the first equality from (\ref{1:cala:15}) is known in the general theory of Fourier transforms, see \cite{stein}, but in the context of our paper they are new and give new point of view on the process of the regularization.

Indeed, the last chain of equalities plays an important role in infrared decompositions of Green's functions, because they allow us to combine several types of the representations of the standard fundamental solution. As we will see this in the following sections, such trick simplifies a number of calculations. We are going to demonstrate this on two types of Laplace operators, which appear in the pure four-dimensional Yang--Mills theory and the standard two-dimensional Sigma-model.

\section{Further calculations}
\label{1:sec:fur}

Now we want to study an application of the regularization to some special operators. First of all, we need to remember some basic concepts of the infrared, near the diagonal, expansion of a Green's function using the corresponding Seeley--DeWitt (sometimes, they are named after Hadamard, Minakshisundaram \cite{111}, and Gilkey \cite{32}) coefficients, see \cite{110,1000}.

Let $L$ be a Laplace-type operator on $\mathbb{R}^d$ with smooth coefficients. Moreover, we concretize its form using the following local formula
\begin{equation}\label{1:cal:1}
L(x)=-D_{x_\mu}D_{x^\mu}-v(x)\,\,\,\mbox{for all}\,\,\,x\in\mathbb{R}^d,
\end{equation}
where $D_{x^\mu}=\partial_{x^\mu}+B_\mu(x)$ is a covariant derivative, $B_\mu(x)$ are components of an 1-form connection, and $v(x)$ is an auxiliary potential. We note that the last objects can be non-commutative between each other.

In this case the Green's function $G(x,y)$, which solves the equation $L(x)G(x,y)=\delta(x-y)$ and satisfy appropriate boundary condition (decreasing at infinity, for example), can be expanded near the diagonal, where $|x-y|$ is bounded by some small fixed parameter, in the following form
\begin{equation}\label{1:cal:2}
G(x,y)=\sum_{n=0}^{+\infty}R_n(x-y)a_n(x,y)+
PS(x,y).
\end{equation}
Here, $a_n(x,y)$ are the Seeley--DeWitt coefficients, which are constructed using the heat kernel method and can be found in the works \cite{Shore,30,31,32,33}. They satisfy the following recurrence relations
\begin{equation}\label{1:cal:3}
(x-y)^\mu D_{x^\mu}a_0(x,y)=0,\,\,\,
a_0(y,y)=1,\,\,\,
\big(n+(x-y)^\mu D_{x^\mu}\big)a_n(x,y)=-L(x)a_{n-1}(x,y),\,\,\,n>0,
\end{equation}
and can be defined by them.
Then, the functions $R_n(x-y)$ do not depend on the connection components and the potential. Their smoothness gets better with increasing of their order number. The first one has the worst smoothness, because it leads to the $\delta$-functional according to the formula $-\partial_{x_\mu}\partial_{x^\mu}R_0(x)=\delta(x)$. Further, the part $PS(x,y)$ is smooth and it includes the global information about the operator and the boundary condition, while the Seeley--DeWitt coefficients are local by construction.

Hence, the non-smooth behavior of the Green's function is included into the $R_n$-functions. We do not write an appropriate  recurrence relations for them, because their definitions can be varied. Indeed, we can subtract a smooth part from the sum and add it to $PS(x,y)$. Anyway, in every example we give explicit formulae.

In the next calculations, we want to demonstrate how to apply the process of regularization in some popular situations, such as the pure four-dimensional Yang--Mills theory and the two-dimensional Sigma-model. 

Additionally, we note that the next operators will have an auxiliary small parameter $s\in\mathbb{R}_+$, such that the connection components and the potential would be small, and, moreover, they can be excluded by the transition $s\to+0$.

\subsection{Pure 4-D Yang--Mills theory} Let $G$ be a compact semisimple Lie group, and $\mathfrak{g}$ is its Lie algebra. Then, let $t^a$ be the generators of the algebra $\mathfrak{g}$, where $a=1,\ldots,\dim\mathfrak{g}$,
such that the relations hold
\begin{equation}\label{constdef}
[t^a,t^b]=f^{abc}t^c,\,\,\,\,\,\,\mathrm{tr}(t^at^b)=-2\delta^{ab},
\end{equation}
where $f^{abc}$ are antisymmetric structure constants for $\mathfrak{g}$, and '$\mathrm{tr}$' is the Killing form. In this case the connection component $B_\mu(x)$ mentioned above can be expanded as $B_\mu^a(x)t^a$, where, as usual, we mean the standard summation by the repeating indices. Then, after introducing the components of the field strength tensor for $x\in\mathbb{R}^4$ in the form
\begin{equation*}
F_{\mu\nu}^{\phantom{a}}(x)=F_{\mu\nu}^a(x)t^a,\,\,\,\mbox{where}\,\,\,
F_{\mu\nu}^a(x)=\partial_{x^\mu}^{}B_\nu^a(x)-\partial_{x^\nu}^{}B_\mu^a(x)+f^{abc}B_\mu^b(x)B_\nu^c(x),
\end{equation*}
we write out two Laplace-type operators
\begin{equation}\label{oper}
M_0^{ab}(x)=-D_{x_{\mu}}^{ae}D_{x^{\mu}}^{eb},\,\,\,
M_{1\mu\nu}^{\,\,\,ab}(x)=M_0^{ab}(x)\delta_{\mu\nu}^{}-2f^{acb}F_{\mu\nu}^c(x),
\end{equation}
where the matrix components $D^{ab}_{x^{\mu}}$ of the covariant derivative $D^{\phantom{a}}_{x^{\mu}}$ have the form
\begin{equation}\label{1:cal:4}
D^{ab}_{x^{\mu}}=\partial_{x^{\mu}}^{\phantom{a}}\delta^{ab}+f^{adb}B^d_\mu(x).
\end{equation}
These operators play a crucial role in the perturbative expansions and in the Feynman diagram technique, see \cite{Ivanov-Kharuk-2020,Ivanov-Kharuk-2022,4,3,21,22,23}, because they define the main quadratic forms, for gauge and ghost fields, and propagators, or Green's functions, which are actually constructed blocks of the technique. 

To obtain the results for both operators simultaneously, we introduce the operator of more general type
\begin{equation}\label{1:cal:9}
M_{2\mu\nu}^{\,\,\,ab}(x;\alpha)=M_0^{ab}(x)\delta_{\mu\nu}^{}-2\alpha f^{acb}F_{\mu\nu}^c(x),
\end{equation}
where $\alpha\in\mathbb{R}$. Then, we have two additional relations
\begin{equation}\label{1:cal:10}
M_0^{ab}(x)=\frac{1}{4}M_{2\mu\mu}^{\,\,\,ab}(x;0),\,\,\,
M_{1\mu\nu}^{\,\,\,ab}(x)=M_{2\mu\nu}^{\,\,\,ab}(x;1).
\end{equation}
Green's functions, corresponding to the operators from (\ref{oper}) and (\ref{1:cal:9}), we notate as
\begin{equation}\label{1:cal:11}
G_0(x,y),\,\,\,G_{1\mu\nu}(x,y),\,\,\,\mbox{and}\,\,\,
G_{2\mu\nu}(x,y;\alpha),\,\,\mbox{respectively}.
\end{equation}

One of the ways to investigate the quantum Yang--Mills theory relates to the renormalization theory, namely, to studying its divergences. Fortunately, see papers \cite{Ivanov-Kharuk-2020,Ivanov-Kharuk-2022,13}, form of the divergent part has a special view, which allows us to consider a simplified version of the connection components. They are equal to
\begin{equation}\label{1:cal:5}
B^a_{\mu}(x)=\frac{s}{2}x^\nu \xi^a_{\nu\mu},\,\,\,\mbox{where}\,\,\,
\big(\xi^a\big)_{\mu\nu}=\frac{1}{\sqrt{8\dim\mathfrak{g}}}
\begin{pmatrix}
0&1&0&1\\
-1&0&1&0\\
0&-1&0&1\\
-1&0&-1&0
\end{pmatrix}\,\,\,
\mbox{for all}\,\,\,a\in\{1,\ldots,\dim\mathfrak{g}\},
\end{equation}
and $s$ is a small auxiliary positive parameter mentioned above.
The corresponding components $F_{\mu\nu}^a$ are equal to $s\xi^a_{\mu\nu}$.
It is quite convenient that we have commutativity and the following two normalization properties
\begin{equation}\label{1:cal:6}
\sum_{a=1}^{\dim\mathfrak{g}}
\sum_{\mu,\nu=1}^4
\xi_{\mu\nu}^a\xi_{\mu\nu}^a=1,\,\,\,\mbox{and}\,\,\,\,
\sum_{\sigma=1}^{4}
\xi_{\mu\sigma}^a\xi_{\sigma\nu}^b=-\frac{2\delta_{\mu\nu}}{8\dim\mathfrak{g}}\,\,\,\mbox{for all}\,\,\,a,b,\mu,\nu.
\end{equation}
Actually, this parameter can be used in the asymptotic expansion instead of $x-y$, if the last difference is from a finite neighborhood of the diagonal $x=y$. In both situations every term of the asymptotics consists of a finite number of subterms, and the transition from one asymptotics to another can be achieved by an explicit resummation.

Using the last simplifications, we can write out the first three Seeley--DeWitt coefficients for the Green's function $G_{2\mu\nu}(x,y;\alpha)$ according to decomposition (\ref{1:cal:2}) as
\begin{equation}\label{1:cal:12}
a_{0\mu\nu}(x,y)=\Phi(x,y)\delta_{\mu\nu},\,\,\,
a_{1\mu\nu}(x,y)=\Phi(x,y)\bigg(2\alpha s\xi_{\mu\nu}+
\frac{s^2\delta_{\mu\nu}}{12}(x-y)^{\sigma\rho}\xi_{\sigma\beta}\xi_{\rho\beta}\bigg),
\end{equation}
\begin{equation}\label{1:cal:14}
a_{2\mu\nu}(x,y)=s^2\Phi(x,y)\bigg(
\frac{\delta_{\mu\nu}}{12}\xi_{\rho\beta}\xi_{\rho\beta}+
2\alpha^2\xi_{\mu\sigma}\xi_{\sigma\nu}\bigg)+o\big(s^2\big),
\end{equation}
where, in the last equality, we have used that $s\to+0$, $\xi_{\mu\nu}$ denotes the $(\dim\mathfrak{g}\times \dim\mathfrak{g})$ matrix with the elements $f^{abc}\xi_{\mu\nu}^b$, and
\begin{equation}\label{1:cal:15}
\Phi(x,y)=\exp\big(x_\mu\xi_{\mu\nu}y_\nu/2\big).
\end{equation}

Then, for convenience, we make a transformation to the Fock--Schwinger gauge, see \cite{13,33}. So, we define the following auxiliary operator
\begin{equation}\label{1:cal:8}
\hat{M}_{2\mu\nu}(x-y;\alpha)=\Phi(y,x)M_{2\mu\nu}(x;\alpha)\Phi(x,y).
\end{equation}
The explicit calculation gives
\begin{equation}\label{1:cal:13}
\hat{M}_{2\mu\nu}(x-y;\alpha)=-\delta_{\mu\nu}\bigg(
\partial_{x_\mu}\partial_{x^\mu}+s(x-y)^\sigma\xi_{\sigma\rho}
\partial_{x_\rho}+\frac{s^2}{4}(x-y)^{\sigma\rho}\xi_{\sigma\beta}\xi_{\rho\beta}
\bigg)-2s\alpha\xi_{\mu\nu}.
\end{equation}

In the case, when $s\to+0$ and $z=x-y$ is bounded, the Green's function for the gauge transformed operator can be written as, see \cite{1000,33,15,psi},
\begin{equation}\label{1:cal:7}
\hat{G}_{2\mu\nu}(z;\alpha)=R_0(z)\delta_{\mu\nu}+
2\alpha s\xi_{\mu\nu}R_1(z)-
\frac{\alpha^2s^2\delta_{\mu\nu}\xi_{\sigma\beta}\xi_{\sigma\beta}}{2}R_2(z)
-\frac{s^2\delta_{\mu\nu}\xi_{\sigma\beta}\xi_{\sigma\beta}|z|^2(1-2\alpha^2)}{2^9\pi^2}+o\big(s^2\big),
\end{equation}
where we have used the relation $\xi_{\mu\sigma}\xi_{\sigma\nu}=-\delta_{\mu\nu}\xi_{\rho\beta}\xi_{\rho\beta}/4$, and
\begin{equation}\label{SD10}
R_0(x)=\frac{1}{4\pi^2|x|^2},\,\,\,
R_1(x)=-\frac{\ln(|x|^2\mu^2)}{16\pi^2},\,\,\,
R_2(x)=\frac{|x|^2\big(\ln(|x|^2\mu^2)-1\big)}{64\pi^2},
\end{equation}
which satisfy the following equations
\begin{equation}\label{smcut6}
-\partial_{x_\mu}\partial_{x^\mu}R_1(x)=R_0(x),\,\,\,
-\partial_{x_\mu}\partial_{x^\mu}R_2(x)=2R_1(x)-\frac{1}{16\pi^2}.
\end{equation}

Now we are ready to formulate the main task of the subsection. We want to find an asymptotic expansion, when $s\to+0$ and $x$ is from a finite neighborhood of $x=0$, of the operator
\begin{equation}\label{1:cal:23}
\rho\bigg(\sqrt{\hat{M}_2(x;\alpha)}/\Lambda\bigg)\hat{G}_{2}(x;\alpha),\,\,\,
\mbox{where}\,\,\,
\rho(r)=\frac{2J_1(r)}{r},
\end{equation}
up to the $s^2$. 
We have written $x$ instead of $z=x-y$, see formulae (\ref{1:cal:13}) and (\ref{1:cal:7}), because the operators depend on the difference. Hence, for simplicity, we can study the case $y=0$.
In the last formula, we understand  all the operators as the matrix-valued one of size $(4\dim\mathfrak{g})\times(4\dim\mathfrak{g})$. Our idea means that we need to find the Taylor expansion for the $\rho$-operator up to $s^2$, and then apply this to the Green's function expansion (\ref{1:cal:7}).

Let us write out some useful lemmas and their proofs.

\begin{lemma}\label{1:cal:lem1}
Let $\Lambda>0$ be quite large, $L=\ln\big(\Lambda/\mu\big)$, $A(x)=-\partial_{x_\mu}\partial_{x^\mu}$, and $\rho(\cdot)$ is the function defined in (\ref{1:cal:23}). Then, under the conditions described above, we have the following equalities
\begin{equation}\label{1:cal:20}
\rho\Big(\sqrt{A(x)}/\Lambda\Big)\big(|x|^2+a\big)=|x|^2+\Lambda^{-2}+a,
\end{equation}
and
\begin{equation}\label{1:cal:19}
\rho\Big(\sqrt{A(x)}/\Lambda\Big)R_i(x)=\tilde{R}_i(x),
\end{equation}
where $a\in\mathbb{R}$, $i=0,1,2$, and $\tilde{R}_i(x)$ is defined as
\begin{equation}\label{1:cal:16}
\tilde{R}_0(x)=\frac{1}{4\pi^2}
\begin{cases}
|x|^{-2},&|x|>1/\Lambda;\\
\,\,\,\Lambda^2,&|x|\leqslant1/\Lambda,
\end{cases}
\end{equation}
\begin{equation}\label{1:cal:17}
\tilde{R}_1(x)=\frac{1}{4\pi^2}
\begin{cases}
-\frac{1}{4}\ln(|x|^2\mu^2)-\frac{1}{8}|x|^{-2}\Lambda^{-2},&|x|>1/\Lambda;\\
\,\,\,\,\,\,\,\,\,\,\,\,\,\,\,\,\,
\frac{1}{2}L-\frac{1}{8}|x|^2\Lambda^2,&|x|\leqslant1/\Lambda,
\end{cases}
\end{equation}
\begin{equation}\label{1:cal:18}
\tilde{R}_2(x)=\frac{1}{4\pi^2}
\begin{cases}
\frac{1}{16}|x|^2\big(\ln(|x|^2\mu^2)-1\big)
+\frac{1}{16}\Lambda^{-2}\ln(|x|^2\mu^2)+\frac{1}{96}|x|^{-2}\Lambda^{-4}
+\frac{1}{32}\Lambda^{-2}
,&|x|>1/\Lambda;\\
\,\,\,\,\,\,\,\,\,\,\,\,\,\,\,\,\,\,\,\,\,\,\,
-\frac{1}{8}\Lambda^{-2}L-\frac{1}{8}|x|^2L+\frac{1}{96}|x|^4\Lambda^2+\frac{1}{32}|x|^2-\frac{1}{16}\Lambda^{-2}
,&|x|\leqslant1/\Lambda.
\end{cases}
\end{equation}
\end{lemma}
\begin{proof}
The first relation can be obtained quite easily using the series representation for the function $\rho(\cdot)$ and the fact that we can analyze only the first two terms, because $|x|^2+a$ is from the kernel of $A^k(x)$ for $k>1$. So, we move on to relation (\ref{1:cal:19}).

Actually, we need to verify the relation only for $i=1,2$, because the case $i=0$ was studied in the previous section. Indeed, the function $R_0(x)$ coincides with the function $G(|x|)$ from (\ref{Green}), when $d=4$. Hence, we can use all the formulae from Section \ref{1:sec:gener}, and, in particular, formula (\ref{1:cala:15}), from which we obtain equality (\ref{1:cal:16}), because $\tilde{R}_0(x)=G(|x|_\Lambda)$.

Then, to simplify calculations, we need to note some useful points. Firstly, the operator $A(x)$ commutes with $\rho\Big(\sqrt{A(x)}/\Lambda\Big)$. Hence, we can expand relations (\ref{smcut6}) on the functions $\tilde{R}_i(x)$. Secondly, the function $\tilde{R}_0(x)$ is continuous. Therefore, the functions $\tilde{R}_1(x)$ and $\tilde{R}_2(x)$ should have the following, additional to (\ref{smcut6}), properties
\begin{equation}
\label{1:cal:21}
\tilde{R}_i(x)\Big|_{|x|=1/\Lambda-0}=
\tilde{R}_i(x)\Big|_{|x|=1/\Lambda+0},\,\,\,
\partial_{x_{\mu}}\tilde{R}_i(x)\Big|_{|x|=1/\Lambda-0}=
\partial_{x_{\mu}}\tilde{R}_i(x)\Big|_{|x|=1/\Lambda+0},
\end{equation}
where $i=1,2$. It is quite easy to verify that the suggested functions satisfy all the conditions mentioned above.

However, the last conditions do not fix one arbitrariness, shift by a constant. To fix this, it is convenient to use the integral representation for $\rho$-operator. Additional relations are
\begin{equation}
\label{1:cal:22}
\rho\Big(\sqrt{A(x)}/\Lambda\Big)R_i(x)\Big|_{x=0}=R_i(x)\Big|_{|x|=1/\Lambda},
\end{equation}
where $i=1,2,3$. Substituting the functions $\tilde{R}_i(x)$ in the right hand side, we are convinced of the validity of the last relations. Hence, we get the main statement of the lemma.
\end{proof}

Further, let us write out the Taylor expansion for the $\rho$-operator up to $s^2$. For that we should note that three parts of the operator $\hat{M}_{2\mu\nu}(x;\alpha)$ have the dependence on the parameter $s$
\begin{equation}\label{1:cal:25}
-s\delta_{\mu\nu}x^\sigma\xi_{\sigma\rho}
\partial_{x_\rho}-\delta_{\mu\nu}\frac{s^2}{4}x^{\sigma\rho}\xi_{\sigma\beta}\xi_{\rho\beta}-2s\alpha\xi_{\mu\nu}.
\end{equation}
Actually, we are interested in the second and the third ones, because the first one gives zero after applying to a spherically-symmetric function. For example, to the function $R_i(x)$. Hence, we can represent the expansion in the form
\begin{multline}\label{1:cal:24}
\rho\bigg(\sqrt{\hat{M}_2(x;\alpha)}/\Lambda\bigg)_{\mu\nu}=
\delta_{\mu\nu}\rho\Big(\sqrt{A(x)}/\Lambda\Big)+s\rho_{1\mu\nu}(x)+
s^2\rho_{2\mu\nu}(x)\\+
\big(s\rho_{3\mu\nu}(x)+s^2\rho_{4\mu\nu}(x)\big)x_\sigma\xi_{\sigma\beta}\partial_{x_\beta}
+o\big(s^3\big),
\end{multline}
where we have two functionals $\rho_{1\mu\nu}(x)$ and $\rho_{2\mu\nu}(x)$, in which we are interested in, and two additional ones $\rho_{3\mu\nu}(x)$ and $\rho_{4\mu\nu}(x)$, which do not contribute in our case and, therefore, are not important. 

The first function can be obtained in a very simple way using the first derivative with respect the parameter $s$. For that let us consider the following chain of equalities
\begin{align}\label{1:cal:26}
\rho_{1\mu\nu}(x)&=\sum_{n=1}^{+\infty}\sum_{k=1}^{n}\frac{\big(-A(x)/4\Lambda^2\big)^{k-1}
\big(2\alpha\xi_{\mu\nu}/4\Lambda^2\big)
\big(-A(x)/4\Lambda^2\big)^{n-k}}{n!\Gamma(n+2)}\\\label{1:cal:27}
&=\frac{\alpha\xi_{\mu\nu}}{2\Lambda^2}\sum_{n=1}^{+\infty}\frac{
\big(-A(x)/4\Lambda^2\big)^{n-1}}{(n-1)!\Gamma(n+2)}=
\frac{2\alpha\xi_{\mu\nu}}{\Lambda^2}
\bigg(\frac{J_2(r)}{r^2}\bigg)\bigg|_{r^2=-A(x)/\Lambda^2}.
\end{align}
Indeed, we can reformulate the last relation using the derivative with respect to the parameter $\Lambda$ as follows
\begin{equation}\label{1:cal:28}
\rho_{1\mu\nu}(x)=
2\alpha\xi_{\mu\nu}
\Lambda^2A^{-1}(x)\frac{\partial}{\partial\Lambda^2}\rho\Big(\sqrt{A(x)}/\Lambda\Big).
\end{equation}

Actually, we are interested in both representations, because it is convenient to derive the general answer with the use of the second one, while the first one helps us to fix an integration parameter.

\begin{lemma}\label{1:cal:lem2}
Under the conditions described above, we have
\begin{equation}\label{1:cal:41}
\rho_{1\mu\nu}(x)R_i(x)=2\alpha\xi_{\mu\nu}\bar{R}_i(x),
\end{equation}
where $i=0,1$, and
\begin{equation}\label{1:cal:29}
\bar{R}_0(x)=\frac{1}{4\pi^2}
\begin{cases}
\,\,\frac{1}{8}|x|^{-2}\Lambda^{-2},&|x|>1/\Lambda;\\
\frac{1}{4}-\frac{1}{8}|x|^2\Lambda^2,&|x|\leqslant1/\Lambda,
\end{cases}
\end{equation}
\begin{equation}\label{1:cal:33}
\bar{R}_1(x)=\frac{1}{4\pi^2}\frac{1}{16\Lambda^2}
\begin{cases}
\,\,
-\frac{1}{2}\ln(|x|^2\mu^2)-\frac{1}{6}|x|^{-2}\Lambda^{-2},&|x|>1/\Lambda;\\
L-\frac{1}{4}-\frac{1}{2}|x|^2\Lambda^2+\frac{1}{12}|x|^4\Lambda^4,&|x|\leqslant1/\Lambda.
\end{cases}
\end{equation}
\end{lemma}
\begin{proof}
Let us start from formula (\ref{1:cal:28}) and apply it to $R_0(x)$. We will proceed in stages, and firstly we use formulae (\ref{1:cal:19}) and (\ref{1:cal:16}) from Lemma \ref{1:cal:lem1}. We have the following chain (without $\alpha\xi_{\mu\nu}$)
\begin{align}\label{1:cal:30}
2\Lambda^2A^{-1}(x)\frac{d}{d\Lambda^2}\rho\Big(\sqrt{A(x)}/\Lambda\Big)R_0(x)&=
2\Lambda^2\frac{d}{d\Lambda^2}A^{-1}(x)\tilde{R}_0(x)=g(\Lambda)+
2\Lambda^2\frac{d}{d\Lambda^2}\tilde{R}_1(x)\\\label{1:cal:31}
&=g(\Lambda)+\frac{1}{4\pi^2}
\begin{cases}
\,\,\frac{1}{4}|x|^{-2}\Lambda^{-2},&|x|>1/\Lambda;\\
\frac{1}{2}-\frac{1}{4}|x|^2\Lambda^2,&|x|\leqslant1/\Lambda,
\end{cases}
\end{align}
where we have introduced an auxiliary function $g(\Lambda)$.
Hence, we need to verify that the function $g(\Lambda)$ is equal to zero. For that we use representation (\ref{1:cal:27}) and the Fourier transform for $R_0(x)$. Using the transition to the spherical coordinates, we have
(without $\alpha\xi_{\mu\nu}$)
\begin{equation}\label{1:cal:32}
\frac{2}{\Lambda^2}\frac{1}{(2\pi)^4}\int_{\mathbb{R}^4}d^4y\,e^{ix_\mu y^\mu}
J_2(|y|/\Lambda)|y|^{-4}\Lambda^2\bigg|_{x=0}=
\frac{1}{4\pi^2}\int_{\mathbb{R}_+}\frac{dr}{r}\,
J_2(r)=\frac{1}{4\pi^2}\frac{1}{2},
\end{equation}
from which the equality $g(\Lambda)=0$ follows.

The second relation of the Lemma can be obtained from the previous one, applying the operator $A^{-1}(x)$, because, due to the commutativity, we have the relation $A(x)\tilde{R}_1(x)=\tilde{R}_0(x)$, see (\ref{smcut6}). Then, satisfying additional conditions of smoothness, as it was performed in Lemma \ref{1:cal:lem1}, see (\ref{1:cal:21}), we obtain the result. Let us additionally note that the constant of integration can be achieved by comparison of asymptotics in the region $|x|\gg1/\Lambda$, where the function is smooth and $A^k(x)R_1(x)=0$ for all $k>1$.
\end{proof}

Now we are ready to find the last function under the consideration, see $\rho_{2\mu\nu}$ in (\ref{1:cal:24}). In the case we need to investigate the second derivative with respect to the parameter $s^2$. Using the series representation, we can write out the following chain of equations

\begin{align}\label{1:cal:34}
\rho_{2\mu\nu}(x)=&\sum_{n=1}^{+\infty}\sum_{k=1}^{n}\frac{\big(-A(x)/4\Lambda^2\big)^{k-1}
\big(|x|^2\delta_{\mu\nu}\xi_{\sigma\beta}\xi_{\sigma\beta}/4^3\Lambda^2\big)
\big(-A(x)/4\Lambda^2\big)^{n-k}}{n!\Gamma(n+2)}\\\label{1:cal:35}
&-\frac{\alpha^2\delta_{\mu\nu}\xi_{\sigma\beta}\xi_{\sigma\beta}}{(4\Lambda^2)^2}\sum_{n=2}^{+\infty}\frac{\big(-A(x)/4\Lambda^2\big)^{n-2}}{n!\Gamma(n+2)}
\sum_{k=1}^{n}(k-1)\\\label{1:cal:36}
=&\frac{|x|^2\delta_{\mu\nu}\xi_{\sigma\beta}\xi_{\sigma\beta}}{2^6\Lambda^2}
\sum_{n=1}^{+\infty}\frac{\big(-A(x)/4\Lambda^2\big)^{n-1}}{(n-1)!\Gamma(n+2)}+
\frac{\delta_{\mu\nu}\xi_{\sigma\beta}\xi_{\sigma\beta}x^\eta\partial_{x_\eta}}{2^7\Lambda^4}
\sum_{n=2}^{+\infty}\frac{\big(-A(x)/4\Lambda^2\big)^{n-2}}{(n-2)!\Gamma(n+2)}\\\label{1:cal:37}
&+\frac{\delta_{\mu\nu}\xi_{\sigma\beta}\xi_{\sigma\beta}}{2^63\Lambda^4}
\sum_{n=2}^{+\infty}\frac{\big(-A(x)/4\Lambda^2\big)^{n-2}}{(n-2)!\Gamma(n+1)}-
\frac{\alpha^2\delta_{\mu\nu}\xi_{\sigma\beta}\xi_{\sigma\beta}}{2^5\Lambda^4}\sum_{n=2}^{+\infty}\frac{\big(-A(x)/4\Lambda^2\big)^{n-2}}{(n-2)!\Gamma(n+2)},
\end{align}
where we have used the following commutation relation
\begin{equation}\label{1:cal:38}
\big(-A(x)\big)^{k-1}|x|^2=|x|^2\big(-A(x)\big)^{k-1}+
4(k-1)x^\eta\partial_{x_\eta}\big(-A(x)\big)^{k-2}+4k(k-1)\big(-A(x)\big)^{k-2},
\end{equation}
and two auxiliary equalities
\begin{equation}\label{1:cal:39}
\sum_{k=1}^n(k-1)=\frac{n(n-1)}{2},\,\,\,
\sum_{k=1}^nk(k-1)=\frac{(n+1)n(n-1)}{3}.
\end{equation}
Then, we need to apply the operator sums from (\ref{1:cal:36}) and (\ref{1:cal:37}) to the function $R_0(x)$. But, firstly, we rewrite the terms in the operator form, preserving their order, as
\begin{equation}\label{1:cal:40}
\rho_{2\mu\nu}(x)
=\frac{\delta_{\mu\nu}\xi_{\sigma\beta}\xi_{\sigma\beta}}{2^5}
\bigg(\frac{|x|^2}{2\Lambda^2}+\bigg(x^\eta\partial_{x_\eta}-\frac{2}{3}
\Lambda^4\frac{\partial}{\partial\Lambda^2}\Lambda^{-2}-4\alpha^2\bigg)
A^{-1}(x)
\frac{\partial}{\partial\Lambda^2}\bigg)\sum_{n=1}^{+\infty}\frac{\big(-A(x)/4\Lambda^2\big)^{n-1}}{(n-1)!\Gamma(n+2)}.
\end{equation}

\begin{lemma}\label{1:cal:lem3}	
Under the conditions described above, we have
\begin{equation}\label{1:cal:42}
\rho_{2\mu\nu}(x)R_{0}(x)=\frac{\delta_{\mu\nu}\xi_{\sigma\beta}\xi_{\sigma\beta}}{2^5}
\Bigg(2|x|^2\bar{R}_0(x)
+\frac{1}{\Lambda^2}\hat{R}_0(x)-
\frac{\alpha^2}{\Lambda^2}
\bigg(\frac{4}{3}\bar{R}_0(x)+\hat{R}_0(x)\bigg)\Bigg),
\end{equation}
where
\begin{equation}\label{1:cal:50}
\hat{R}_0(x)=\frac{1}{4\pi^2}\begin{cases}
\,\,\,\,\,\,\,\,\,\,\,\,\,\,\,\,\,\,\,\,\,\,\,\,\,\,
0,&|x|>1/\Lambda;\\
\frac{1}{6}-\frac{1}{3}|x|^2\Lambda^2+\frac{1}{6}|x|^4\Lambda^4,&|x|\leqslant1/\Lambda.
\end{cases}
\end{equation}
\end{lemma}
\begin{proof}
We are going to use some useful formulae derived above. Firstly, using relations 
(\ref{1:cal:27}) and (\ref{1:cal:41})--(\ref{1:cal:33}), we get
\begin{equation}\label{1:cal:43}
\sum_{n=1}^{+\infty}\frac{\big(-A(x)/4\Lambda^2\big)^{n-1}}{(n-1)!\Gamma(n+2)}R_0(x)=
4\Lambda^2\bar{R}_0(x),
\end{equation}
\begin{align}\label{1:cal:44}
A^{-1}(x)\frac{\partial}{\partial\Lambda^2}
\sum_{n=1}^{+\infty}\frac{\big(-A(x)/4\Lambda^2\big)^{n-1}}{(n-1)!\Gamma(n+2)}R_0(x)&=
g(\Lambda)+\frac{\partial}{\partial\Lambda^2}4\Lambda^2\bar{R}_1(x)\\\label{1:cal:46}
&=g(\Lambda)+\frac{1}{4\pi^2}\frac{1}{4}
\begin{cases}
\,\,\,\,\,\,\,\,\,\,\,\,\,\,\,\,\,
\frac{1}{6}|x|^{-2}\Lambda^{-4},&|x|>1/\Lambda;\\
\frac{1}{2}\Lambda^{-2}-\frac{1}{2}|x|^2+\frac{1}{6}|x|^4\Lambda^2,&|x|\leqslant1/\Lambda,
\end{cases}
\end{align}
where $g(\Lambda)$ is a constant of integration, which is actually equal to zero. It is easy to verify using the Fourier transform at zero. Indeed, the left hand side of (\ref{1:cal:44}) at $x=0$ equals
\begin{align}\label{1:cal:45}
\frac{1}{4\Lambda^4}
\sum_{n=2}^{+\infty}\frac{\big(-A(x)/4\Lambda^2\big)^{n-2}}{(n-2)!\Gamma(n+2)}R_0(x)\Bigg|_{x=0}&=\frac{1}{4\Lambda^2}\frac{1}{(2\pi)^4}\int_{\mathbb{R}^4}d^4y\,\frac{8J_3(|y|/\Lambda)}{|y|^5}\Lambda^3\\\label{1:cal:48}&=
\frac{1}{4\pi^2\Lambda^2}\int_{\mathbb{R}_+}dr\,\frac{J_3(r)}{r^2}=
\frac{1}{4\pi^2}\frac{1}{8\Lambda^2},
\end{align}
where we have used the relation
\begin{equation}\label{1:cal:47}
\sum_{n=2}^{+\infty}\frac{\big(-r^2/4\big)^{n-2}}{(n-2)!\Gamma(n+2)}=\frac{8J_3(r)}{r^3},
\end{equation}
and the transition to the spherical coordinates. Then, applying the differential operator from (\ref{1:cal:40}) to (\ref{1:cal:46}), we get the final result.
\end{proof}

Now we are ready to formulate the final result of this subsection for the Green's functioned mentioned in (\ref{1:cal:11}), (\ref{1:cal:7}), and (\ref{1:cal:23}).
\begin{theorem}\label{1:cal:th1}
Let $x,y\in\mathbb{R}^4$, $z=x-y$, $\Lambda\gg 0$, and $s$ is a quite small positive number. Then, under the conditions described above, we have
\begin{equation}\label{1:cal:55}
\rho\bigg(\sqrt{M_2(x;\alpha)}/\Lambda\bigg)_{\mu\sigma}G_{2\sigma\nu}^{\phantom{r}}(x,y;\alpha)=
\Phi(x,y)\hat{G}_{2\mu\nu}^{\Lambda}(z;\alpha),
\end{equation}
where
\begin{align}\label{1:cal:51}
\hat{G}_{2\mu\nu}^{\Lambda}(z;\alpha)&=\tilde{R}_0(z)\delta_{\mu\nu}+
2\alpha s\xi_{\mu\nu}\Big(\bar{R}_0(z)+\tilde{R}_1(z)\Big)\\\label{1:cal:53}
&+s^2\delta_{\mu\nu}\xi_{\sigma\beta}\xi_{\sigma\beta}
\Bigg(\frac{3|z|^2\Lambda^2-2\alpha^2}{2^43\Lambda^2}\bar{R}_0(z)
+\frac{1-\alpha^2}{2^5\Lambda^2}\hat{R}_0(z)-\alpha^2\bar{R}_1(z)\\\label{1:cal:52}
&\,\,\,\,\,\,\,\,\,\,\,\,\,\,\,\,\,\,\,\,\,\,\,\,\,\,\,\,\,\,\,\,\,\,\,\,\,\,\,\,\,
\,\,\,\,\,\,\,\,\,\,\,\,\,\,\,\,\,\,\,\,\,\,\,\,\,\,\,\,-
\frac{\alpha^2}{2}\tilde{R}_2(z)-\frac{\big(|z|^2+\Lambda^{-2}\big)(1-2\alpha^2)}{2^9\pi^2}\Bigg)+\mathcal{O}\big(s^3\big).
\end{align}
Moreover, as a particular case of the last formula we have results for both initial operators (\ref{1:cal:10}) and their Green's functions (\ref{1:cal:11})
\begin{equation}\label{1:cal:54}
\rho\bigg(\sqrt{M_0(x)}/\Lambda\bigg)G_{0}^{\phantom{r}}(x,y)=
\frac{1}{4}\Phi(x,y)\hat{G}_{2\mu\mu}^{\Lambda}(z;0),
\end{equation}
\begin{equation}\label{1:cal:56}
\rho\bigg(\sqrt{M_1(x)}/\Lambda\bigg)_{\mu\sigma}G_{1\sigma\nu}^{\phantom{r}}(x,y)=
\Phi(x,y)\hat{G}_{2\mu\nu}^{\Lambda}(z;1).
\end{equation}
\end{theorem}
\begin{proof}
The statements of the theorem follow from acting by operator (\ref{1:cal:24}) on decomposition (\ref{1:cal:7}) and using the corresponding lemmas proved above.
\end{proof}

\subsection{2-D Sigma-model}

Now we move on to the next example, a two-dimensional Sigma-model. Moreover, we consider the simplest case, the classical action of which can be represented in the form, see \cite{sig2,sig1},
\begin{equation}\label{1:cala:1}
S[g]=\frac{1}{\gamma^2}\int_{\mathbb{R}^2}d^2x\,\mathrm{tr}\Big(\big(\partial_{x_\mu}g(x)\big)
\big(\partial_{x^\mu}g^{-1}(x)\big)\Big),
\end{equation}
where $g(x)\in SU(N)$ for $N\in\mathbb{N}$, $\mu\in\{1,2\}$, and $\gamma$ is a coupling constant. 

Investigating the path integral for the classical action and using the background field method lead to the Feynman perturbative expansion, the propagator for which is constructed using the following Laplace-type operator
\begin{equation}\label{1:cala:2}
M^{ab}(x)=-D_{x_{\mu}}^{ab}\partial_{x^\mu}^{\phantom{a}},
\end{equation}
where the covariant derivative is defined by the same formula (\ref{1:cal:4}), but for the two-dimensional case. The connection components take its values in the Lie algebra $\mathfrak{su}(N)$. Of course, we assume that the main relations for the generators from (\ref{constdef}) hold as well.

As it follows from the explicit calculations of divergences, the non-zero contribution gives only the quadratic density $B_\mu^a(x)B_\mu^a(x)$. It means that we can introduce some convenient simplifications. Firstly, we exclude the variable $x$ from the connection components. Hence, we obtain the constant fields, such that $\partial_{x^\mu}B_\nu(x)=0$ for all $\mu$ and $\nu$. Secondly, we add the commutativity, such that the matrices $B_\mu$, components of which are $f^{abc}B^b_\mu$, commute with each other $[B_\mu,B_\nu]=0$ for all $\mu$ and $\nu$. Thirdly, we want the connection components to be small. To satisfy the last condition, we multiply them by an auxiliary small parameter $s>0$.

It is easy to verify that the first two conditions also can be satisfied easily. For example, we can take the connection components in the form $B_\mu^a=1$ for all $a$ and $\mu$.

Let us define two auxiliary matrix-valued operators
\begin{equation}\label{1:cala:3}
\Phi_0(x-y)=\exp\Big(-s(x-y)^\mu B_\mu/2\Big),\,\,\,
m=-\frac{1}{4}B_\mu B_\mu.
\end{equation}

Hence, after performing the transition to the Fock--Schwinger gauge condition, as it was made in (\ref{1:cal:8}), we obtain the following new operator under the study
\begin{equation}\label{1:cala:4}
\hat{M}(x)=\Phi_0(y-x)M(x)\Phi_0(x-y)=
-\partial_{x_\mu}\partial_{x^\mu}-s^2m.
\end{equation}
The Green's functions, corresponding to the last operators, have the following expansions,
see \cite{33,15,psi},
\begin{equation}\label{1:cala:5}
\hat{G}(x-y)=R_0(x-y)+s^2mR_1(x-y)+o\big(s^2\big),\,\,\,
G(x-y)=\Phi_0(x,y)\hat{G}(x-y),
\end{equation}
where we have used the two-dimensional analogies for the functions
\begin{equation}\label{1:cala:6}
R_0(x)=-\frac{1}{4\pi}\ln(|x|^2\mu^2),\,\,\,
R_1(x)=\frac{|x|^2\big(\ln(|x|^2\mu^2)-2\big)}{16\pi},
\end{equation}
which satisfies the following equations
\begin{equation}\label{1:cala:7}
A(x)R_0(x)=\delta(x),\,\,\,
A(x)R_1(x)=R_0(x),\,\,\,A(x)=-\partial_{x_\mu}\partial_{x^\mu}.
\end{equation}

Now we are ready to formulate and prove the main result of this subsection.

\begin{theorem}\label{1:cal:th2}
Let $x,y\in\mathbb{R}^2$, $\Lambda\gg 0$, $s\to+0$, and $\rho(r)=J_0(|r|)$. Then, under the conditions described above, we have
\begin{equation}\label{1:cala:8}
\rho\bigg(\sqrt{\hat{M}(x)}/\Lambda\bigg)\hat{G}(x-y)=
\tilde{R}_0(x-y)+s^2m\Big(\tilde{R}_1(x-y)+\hat{R}_0(x-y)\Big)+o\big(s^2\big),\,\,\,
\end{equation}
and
\begin{equation}\label{1:cala:9}
\rho\Big(\sqrt{M(x)}/\Lambda\Big)G(x-y)=\Phi_0(x-y)
\rho\bigg(\sqrt{\hat{M}(x)}/\Lambda\bigg)\hat{G}(x-y)
,
\end{equation}
where
\begin{equation}\label{1:cala:10}
\tilde{R}_0(x)=
\rho\Big(\sqrt{A(x)}/\Lambda\Big)R_0(x)=\frac{1}{4\pi}
\begin{cases}
-\ln(|x|^2\mu^2),&|x|>1/\Lambda;\\
\,\,\,\,\,\,\,\,\,\,\,\,\,\,
2L,&|x|\leqslant1/\Lambda,
\end{cases}
\end{equation}
\begin{equation}\label{1:cala:11}
\tilde{R}_1(x)=
\rho\Big(\sqrt{A(x)}/\Lambda\Big)R_1(x)
=\frac{1}{4\pi}
\begin{cases}
\frac{1}{4}|x|^2\big(\ln(|x|^2\mu^2)-2\big)+
\frac{1}{4}\Lambda^{-2}\ln(|x|^2\mu^2),&|x|>1/\Lambda;\\
\,\,\,\,\,\,\,\,\,\,\,\,\,\,\,\,\,\,\,
-\frac{1}{2}(L+1)\Lambda^{-2}-\frac{1}{2}L|x|^2,&|x|\leqslant1/\Lambda,
\end{cases}
\end{equation}
\begin{equation}\label{1:cala:12}
\hat{R}_0(x)=m^{-1}\frac{\partial}{\partial s^2}\bigg|_{s=0}
\rho\Big(\sqrt{A(x)}/\Lambda\Big)R_0(x)=\frac{1}{4\pi}
\begin{cases}
\,\,\,\,\,\,\,\,-\frac{1}{4}\Lambda^{-2}\ln(|x|^2\mu^2),&|x|>1/\Lambda;\\
\frac{1}{2}L\Lambda^{-2}+\frac{1}{4}(\Lambda^{-2}-|x|^2),&|x|\leqslant1/\Lambda.
\end{cases}
\end{equation}
\end{theorem}
\begin{proof} The statement from (\ref{1:cala:9}) follows from formula (\ref{1:cala:8}) with the use of the second equality from (\ref{1:cala:5}). Hence, we need to verify only (\ref{1:cala:8}), which, actually, can be obtained by the same steps that were performed in the previous section. Let us proceed in stages.
	
Firslty, formula (\ref{1:cala:10}) is the consequence of Lemmas \ref{fur1}, \ref{fur3}, and \ref{fur2}, because in the two-dimensional space $R_0(x)$ coincides with $G(|x|)$ from (\ref{Green}), and $\tilde{R}_0(x)$ coincides with $G(|x|_\Lambda)$.

Secondly, formula (\ref{1:cala:11}) can be obtained using the second equality from (\ref{1:cala:7}) and an explicit integration with the use of the following smoothness conditions
\begin{equation}\label{1:cala:13}
\tilde{R}_1(x)\Big|_{|x|=1/\Lambda-0}=\tilde{R}_1(x)\Big|_{|x|=1/\Lambda+0},\,\,\,
\partial_{x_\mu}\tilde{R}_1(x)\Big|_{|x|=1/\Lambda-0}=\partial_{x_\mu}\tilde{R}_1(x)\Big|_{|x|=1/\Lambda+0}.
\end{equation}
Then, fixing the integration constant by comparing asymptotics in the smooth region $|x|\gg 1/\Lambda$, we get the result formulated above. Finally, repeating all the last calculations with the use of (\ref{1:cal:28}), where we have substituted $s^2m$ instead of $2s\alpha\xi_{\mu\nu}$, we get the last statement of the theorem.
\end{proof}

\section{Conclusion}
\label{1:sec:conc}

In this paper, we have studied the explicit cutoff regularization in the coordinate representation, see (\ref{1:int3}). We constructed it in such a way that three additional conditions would be satisfied: the spectral representation, the homogenization (special type of an integral representation), and the covariance, if it exists before the regularization introduced. All these restrictions were achieved on the example of the standard Laplace operator $-\partial_{x_\mu}\partial_{x^\mu}$ for an arbitrary dimension value, see formula (\ref{1:cala:15}), and on the example of two auxiliary operators, appearing in the four-dimensional Yang--Mills theory and in the two-dimensional Sigma-model, see Section \ref{1:sec:fur} and Theorems \ref{1:cal:th1} and \ref{1:cal:th2}.

Of course, we need to give some comments about a behavior of (\ref{1:cal:51})--(\ref{1:cal:52}) and (\ref{1:cala:8}) with respect the auxiliary parameter $s$. Indeed, we have not proven any result about their convergence, so we supposed that they are asymptotic. Finding an operator norm, with respect to which the convergence would be valid, is a separate task and this is out of the scope of our paper.

In our opinion, these results can be useful in  multi-loop calculations in the case of the pure four-dimensional Yang--Mills theory and the two-dimensional Sigma-model. The answers described above contain all the necessary terms of the infrared expansion, needed for studying of divergent parts.

Let us note that at the very beginning of the paper, we have fixed the view of the regularization under study, see (\ref{1:cut1}). This restriction is very strong, because it controls the form of the spectral function $\rho(r)$, see formulae (\ref{1:int1}) and (\ref{ffa1}). Actually, we can abandon this restriction and define the spectral function in an arbitrary convenient way, not for the explicit cutoff regularization. For example, we can take $\rho(r)=\exp(-r^2)$, which does not depend on the dimension value. Anyway, each particular case should be studied separately, and this problem is the subject of further research.

Additionally, it would be useful to expand the result on the case of potentials $v(x)$, depending on the variable $x$, see (\ref{1:cal:1}). Such generalizations make calculations, appearing in the multi-dimensional cases, to be more convenient and simple. Fortunately, the main standard models can be investigated using the methods described above. For example, we can study a $\phi^3$-model in the six-dimensional space or a $\phi^4$-model in the four-dimensional space, which lead to the Laplace operator $-\partial_{x_\mu}\partial_{x^\mu}-v(x)$ with smooth potential, depending on $x$. In such situations, we can consider only several terms from the Taylor expansion of $v(x)$.

\paragraph{Acknowledgements.}
The author is grateful to Natalia Kharuk for thoughtful reading of the manuscript, comments, and collaboration on related topics.
This research is supported by the Ministry of Science and Higher Education of the Russian Federation, agreement  075-15-2022-289.


\begin{thebibliography}{99}
\bibitem{9}
C. Itzykson, J. B. Zuber, \textit{Quantum Field Theory}, Mcgraw-hill, New York, 1--705 (1980)
\bibitem{10}
M. E. Peskin, D. V. Schroeder, \textit{An Introduction to Quantum Field Theory}, Addison-Wesley, 1--868 (1995)
\bibitem{6}
J. C. Collins, \textit{Renormalization: An Introduction to Renormalization, the Renormalization Group and the Operator-Product Expansion}, Cambridge University Press, 1--392 (1986)
\bibitem{7}
O. I. Zavialov, \textit{Renormalized quantum field theory}, Kluwer Academic Publishers, Dodrecht, Boston, 1--524 (1990)
\bibitem{105}
D. I. Kazakov, \textit{Radiative Corrections, Divergences, Regularization, Renormalization, Renormalization Group and All That in Examples in Quantum Field Theory}, arXiv:0901.2208 [hep-ph] (2009)
	
	
\bibitem{Ivanov-Kharuk-2019}
A. V. Ivanov, N. V. Kharuk, \textit{Quantum equation of motion and two-loop cutoff renormalization for $\phi^3$ model}, Questions of quantum field theory and statistical physics. Part 26, Zap. Nauchn. Sem. POMI, \textbf{487}, 151--166 (2019), J. Math. Sci. \textbf{257}, 526--536 (2021) arXiv:2203.04562 [hep-th]

\bibitem{Ivanov-Kharuk-2020}
A. V. Ivanov, N. V. Kharuk, \textit{Two-loop cutoff renormalization of 4-D Yang--Mills effective action}, J. Phys. G: Nucl. Part. Phys. \textbf{48}, 015002 (2020)

\bibitem{Ivanov-Kharuk-2022}
A. V. Ivanov, N. V. Kharuk, \textit{Formula for two-loop divergent part of 4-D Yang--Mills effective action}, arXiv:2203.07131 [hep-th] (2022)

\bibitem{w6}
M. Oleszczuk, \textit{A symmetry-preserving cut-off regularization},
Z. Phys. C, \textbf{64}, 533--538 (1994)
\bibitem{w7}
Sen-Ben Liao, \textit{Operator Cutoff Regularization and Renormalization Group in 
	Yang-Mills Theory}, Phys. Rev. D, \textbf{56}, 5008--5033 (1997)
\bibitem{w8}
G. Cynolter, E. Lendvai, \textit{Cutoff Regularization Method in Gauge Theories},
[arXiv:1509.07407 [hep-ph]] (2015)

\bibitem{Birman}
M. S. Birman, M. Z. Solomjak, \textit{Spectral Theory of Self-Adjoint Operators in Hilbert Space}, D. Reidel Publishing Company, Dordrecht, Holland, 1--302 (1987)

\bibitem{Gelfand-1964}
I. M. Gel'fand, G. E. Shilov, 
\textit{Generalized Functions, Volume 1: Properties and Operations},
AMS Chelsea Publishing \textbf{377}, 1--423 (1964)

\bibitem{Vladimirov-2002}
V. S. Vladimirov, \textit{Methods of the theory of generalized functions}, London,
CRC Press, 1--328 (2002)
	
\bibitem{Vladimirov-1981}
V. S. Vladimirov, \textit{Equations of mathematical physics}, Moscow, Nauka, Fourth edition, 1--512 (1981)
	
\bibitem{DIS-2020}
S. E. Derkachev, A. V. Ivanov, L. A. Shumilov, \textit{Mellin--Barnes transformation for two-loop master-diagrams}, Questions of quantum field theory and statistical physics. Part 27, Zap. Nauchn. Sem. POMI, \textbf{494}, POMI, St. Petersburg, 144--167 (2020)

\bibitem{stein}
E. Stein, G. Weiss, \textit{Introduction to Fourier Analysis on Euclidean Spaces}, Princeton, N.J.: Princeton University Press, 1--312 (1971)

\bibitem{111}
G. W. Gibbons, \textit{Quantum field theory in curved spacetime}, General Relativity, An Einstein Centenary Survey, 639--679 (1979)

\bibitem{32}
P. B. Gilkey, \textit{The spectral geometry of a Riemannian manifold}, J. Differ. Geom. \textbf{10}, 601--618 (1975)

\bibitem{1000}
B. S. DeWitt, \textit{Dynamical Theory of Groups and Fields}, Gordon and Breach, New York, 1--248 (1965)

\bibitem{110}
R. T. Seeley, \textit{Complex powers of an elliptic operator},
Singular Integrals, Proc. Sympos. Pure Math. \textbf{10},
Amer. Math. Soc., 288--307 (1967)

\bibitem{Shore}
G. M. Shore, \textit{Symmetry restoration and the background field method in gauge theories}, Ann. Phys. \textbf{137}(2), 262--305 (1981)

\bibitem{30}
A. O. Barvinsky, G. A. Vilkovisky, \textit{The Generalized Schwinger--Dewitt Technique in Gauge Theories and Quantum Gravity},
Phys. Rept. \textbf{119}, 1--74 (1985)

\bibitem{31}
A. V. Ivanov, \textit{Diagram Technique for the Heat Kernel of the Covariant Laplace Operator}, TMF, \textbf{198}:1, 113--132, (2019); Theoret. and Math. Phys., \textbf{198}:1, 100--117 (2019)
doi:10.1134/S0040577919010070 [arXiv:1905.05455 [hep-th]]

\bibitem{33}
A. V. Ivanov, N. V. Kharuk, \textit{Heat kernel: Proper-time method, Fock--Schwinger gauge, path integral, and Wilson line}, TMF, \textbf{205}:2, 242--261, (2020); Theoret. and Math. Phys., \textbf{205}:2, 1456--1472 (2020) https://doi.org/10.1134/S0040577920110057
\bibitem{4}
L. D. Faddeev, V. Popov, \textit{Feynman Diagrams for Yang--Mills field}, Phys. Lett. B, \textbf{25}, 29--30 (1967)
\bibitem{3}
L. D. Faddeev, A. A. Slavnov, \textit{Gauge Fields: An Introduction to Quantum Theory}, Frontiers in Physics \textbf{83}, Addison-Wesley, 1--236 (1991)
\bibitem{23}
L. D. Faddeev, \textit{Mass in Quantum Yang--Mills theory (comment on a Clay millenium
	problem)}, Bull. Braz. Math. Soc. (N. S.), \textbf{33}:2, 201--212 (2002) arXiv: 0911.1013
	\bibitem{21}
L. D. Faddeev, \textit{Scenario for the renormalization in the 4D Yang--Mills theory}, Int. J. Mod. Phys. A, \textbf{31}, 1630001 (2016)
\bibitem{22}
S. E. Derkachev, A. V. Ivanov, L. D. Faddeev, \textit{Renormalization scenario for the quantum Yang--Mills theory in four-dimensional space--time}, TMF, \textbf{192}:2, 227--234, (2017); Theoret. and Math. Phys., \textbf{192}:2, 1134--1140 (2017) https://doi.org/10.1134/S0040577917080049


\bibitem{13}
J. P. Bornsen, A. E. M. van de Ven, \textit{Three-loop Yang--Mills $\beta$-function via the covariant background field method}, Nucl. Phys. B, \textbf{657}, 257--303 (2003)

\bibitem{15}
M. Lüscher, \textit{Dimensional regularisation in the presence of large background fields}, Ann. Phys., \textbf{142}, 359--392 (1982)
\bibitem{psi}
A. V. Ivanov, N. V. Kharuk, \textit{Two Function Families and Their Application to Hankel Transform of Heat Kernel}, arXiv:2106.00294v1 [math-ph], (2021),
\textit{Special Functions for Heat Kernel Expansion}, arXiv:2106.00294v2 [math-ph] (2022)

\bibitem{sig2}
D. Friedan, \textit{Nonlinear models in $2+\varepsilon$ dimensions}, Ann. Phys. \textbf{163}, 318--419 (1985)
\bibitem{sig1}
A. M. Polyakov, \textit{Gauge Fields and Strings}, London, Taylor and Francis Group, 1--312 (1987)



\end{thebibliography}
\end{document}